\documentclass[11pt,a4paper,reqno]{amsart}%
\usepackage{amsthm,amsmath,amsfonts,amssymb,amsxtra,appendix,bookmark,dsfont,bm,mathrsfs,color}
\usepackage{bbm} 

\usepackage[a4paper,left=40mm,top=25mm,bottom=25mm,right=40mm]{geometry}

\newtheorem{theorem}{Theorem}

\newtheorem{lemma}[theorem]{Lemma}

\theoremstyle{definition}

\theoremstyle{remark}

\DeclareMathOperator{\Tr}{Tr}

\def\bq{\begin{eqnarray}}
\def\eq{\end{eqnarray}}
\def\bqq{\begin{align*}}
\def\eqq{\end{align*}}

\def\eps{\varepsilon}

\newcommand\1{{\ensuremath {\mathds 1} }}

\def\R {\mathbb{R}}

\def\R {\mathbb{R}}

\def\d{{\rm d}}

\title[] {A proof of the Lieb--Thirring inequality via the Besicovitch covering lemma}

\author[P. T. Nam]{Phan Th\`anh Nam} 
\address{LMU Munich, Department of Mathematics, Theresienstrasse 39, D-80333 Munich, Germany} 
\email{nam@math.lmu.de}

\begin{document}

\maketitle

\begin{center}
{\em Dedicated to Professor Duong Minh Duc on the occasion of his 70th birthday}  
\end{center}

\begin{abstract} We give a proof of the Lieb--Thirring inequality on the kinetic energy of orthonormal functions by using a microlocal technique, in which the  uncertainty and exclusion principles are combined through the use of the Besicovitch covering lemma. 

%
\end{abstract}

\date{\today}


\section{Introduction}

The celebrated Lieb-Thirring inequality \cite{LieThi-75,LieThi-76} states that for every trace class operator $0\le \gamma \le 1$ on $L^2(\R^d)$ with density $\rho_\gamma(x)=\gamma(x,x)$, the kinetic energy lower bound 
\bq \label{eq:LT}
\Tr(-\Delta \gamma) \ge K_d \int_{\R^d} \rho_\gamma(x)^{1+\frac{2}{d}} \d x
\eq
holds with a  constant $K_d>0$ depending only on the dimension $d\ge 1$ (in particular, $K_d$ is independent of $\gamma$). Here assuming the spectral decomposition 
\bq \label{eq:gamma-un} \gamma= \sum_{n=1}^\infty \lambda_n |u_n\rangle \langle u_n|
\eq
with $\{u_n\}_{n=1}^\infty$ being orthonormal in $L^2(\R^d)$ and  $0\le \lambda_n\le 1$, we denote  
\bq \label{eq:gamma-rho}
\Tr(-\Delta \gamma)  = \sum_{n=1}^\infty \lambda_n \int_{\R^d} |\nabla u_n|^2, \quad \rho_\gamma(x)=\sum_{n=1}^\infty \lambda_n |u_n(x)|^2. 
\eq

Strictly speaking, the trace class condition $\Tr \gamma=\sum_{n=1}^\infty \lambda_n <\infty$ is not essential. The bound \eqref{eq:LT} still holds for every operator $0\le \gamma\le 1$ provided that $\gamma$ admits a spectral decomposition as in \eqref{eq:gamma-un} such that we can define the kinetic energy $\Tr(-\Delta \gamma)$ and the density $\rho_\gamma$ via \eqref{eq:gamma-rho} (this extension can be obtained easily from a density argument using the Lebesgue monotone convergence theorem). In connection to quantum physics, $\gamma$ usually stands for the one-body density matrix of a quantum state, and the trace class condition means that the number of particles is finite. A  key feature of \eqref{eq:LT} is that the bound is independent of the number of particles, thus allowing us to treat large systems (including even infinite systems via a standard limiting argument).

By some sort of the uncertainty principle (e.g. Sobolev's inequality), it is not difficult to see that
\bq \label{eq:S}
\left( \int_{\R^d} |\nabla u|^2 \right) \left( \int_{\R^d} |u|^2 \right)^{\frac 2 d} \ge K_{\rm So} \int_{\R^d} |u(x)|^{2(1+\frac{2}{d})} \d x, \quad \forall u\in H^1(\R^d).
\eq
Here the optimal constant $K_{\rm So}>0$ depends only on the dimension $d\ge 1$. Consequently, if $\{u_n\}_{n=1}^N$ is a family of {\em normalized} functions in $L^2(\R^d)$, then using \eqref{eq:S} we have
\bq \label{eq:S1}
\sum_{n=1}^N \int_{\R^d} |\nabla u_n|^2 \ge   K_{\rm So} \int_{\R^d} \sum_{n=1}^N |u_n|^{2(1+\frac{2}{d})}   \ge \frac{K_{\rm So} }{N^{2/d}}  \int_{\R^d} \left( \sum_{n=1}^N |u_n|^{2}  \right)^{1+\frac 2 d}. 
\eq 
The optimality of the constant $K_{\rm So} N^{-2/d}$ in \eqref{eq:S1} can be seen easily by considering the special case $u_n=u$ for all $n$. On the other hand, if $\{u_n\}_{n=1}^N$ are {\em orthonormal}, then the Lieb-Thirring inequality \eqref{eq:LT} with $\gamma=\sum_{n=1}^N |u_n\rangle \langle u_n|$ implies that  \eqref{eq:S1} holds with $K_{\rm So} N^{-2/d}$ replaced by a constant $K_d$ depending only on the dimension $d\ge 1$, which is very useful when $N\to \infty$. The orthogonality here essentially corresponds to the condition $0\le \gamma\le 1$ in \eqref{eq:LT}.

As an application, the Lieb--Thirring inequality \eqref{eq:LT} gives a semiclassical estimate for the kinetic energy of fermionic systems, which is a crucial ingredient of the proof of the stability of matter in \cite{LieThi-75}. In this aspect, the constraint $0\le \gamma \le 1$ corresponds to Pauli's exclusion principle of fermions (see e.g.  \cite[Theorem 3.2]{LieSei-10}). Therefore,  \eqref{eq:LT} can be interpreted as an elegant combination of the uncertainty and exclusion principles in quantum mechanics.

Historically, in the original proof of Lieb and Thirring in  \cite{LieThi-75,LieThi-76}, the kinetic bound \eqref{eq:LT} was derived from its dual version 
\bq \label{eq:LT-V}
\Tr[(-\Delta + V(x))_-]  \ge -L_d \int_{\R^d} |V_-(x)|^{1+\frac d 2} \d x
\eq
where $t_-=\min(t,0)$ and 
$$
K_d\left( 1 + \frac 2 d \right) = \left[ L_d \left( 1+ \frac d 2\right) \right]^{-\frac 2 d}.
$$
In this approach, the sum of all negative eigenvalues of the Schr\"odinger operator $-\Delta + V(x)$ is estimated via the number of eigenvalues $\ge 1$ of the
Birman--Schwinger operator  $
\sqrt{|V_-(x)|} (-\Delta +E)^{-1} \sqrt{|V_-(x)|}
$ with $E>0$. This duality approach is powerful since it relates the nonlinear inequality \eqref{eq:LT} to a more linear problem, namely a spectral property of a linear operator.  The spectral estimate \eqref{eq:LT-V} is interesting in its own right and has been generalized in various directions; see \cite{Frank-20} for a recent review with many open questions. 

In this article, we will give a direct proof of the kinetic bound \eqref{eq:LT}, without relying on the duality argument. The first proof in this direction goes back to the work of Eden and Foias in 1991 \cite{EdeFoi-91} for one dimension, which was later extended by Dolbeault, Laptev, and Loss \cite{DolLapLos-08} for higher dimensions. More recently, two other direct proofs were given by Rumin \cite{Rumin-10,Rumin-11} and by Lundholm and Solovej \cite{LunSol-13,LunSol-14}; we refer to \cite{Nam-19} for a review on these two approaches. See also \cite{Sabin-16} for another direct proof. 

We will discuss a simplification of the Lundholm--Solovej approach. The key observation of this approach, which actually goes back to the first proof of the stability of matter by Dyson and Lenard \cite{DysLen-67,DysLen-67b},  is that a local version of the Lieb--Thirring inequality on cubes can be obtained easily by local versions of the uncertainty and exclusion principles. As realized in \cite{LunSol-13,LunSol-14}, it is possible to obtain the global bound  \eqref{eq:LT} by combining such local estimates by choosing an appropriate family of disjoint cubes. A simplification of the Lundholm--Solovej approach was already given by Lundholm--Nam--Portmann \cite{LunNamPor-16} where the local uncertainty and exclusion principles are combined by using a new covering lemma which is inspired by a ``stopping time argument" in harmonic analysis. Our aim here is to provide with a further simplification where the local uncertainty and exclusion principles are combined naturally via the classical {\em Besicovitch covering lemma}, thus making the proof very elementary. We hope that our proof will be not only interesting from the pedagogical point of view, but also helpful for finding new generalizations of \eqref{eq:LT}   in the future. In particular, we expect that the method represented here can be strengthened by coupling with the ideas of using dyadic sets in recent works \cite{KogNam-20,MalNgu-21}.

Note that the Besicovitch covering lemma was already used in Rozenblum's proof of the Cwikel--Lieb--Rozenblum inequality \cite{Cwikel-77,Lieb-76,Rozenblum-76} on the number of negative eigenvalues of Schr\"odinger operators for $d\ge 3$:
\bq \label{eq:CLR-V}
\Tr \Big[ \1_{(-\infty,0)} ( -\Delta + V(x) ) \Big] \le C_d \int_{\R^d} |V_-(x)|^{\frac d 2} \d x. 
\eq
In this approach, using Sobolev's inequality, the linear operator $-\Delta + V(x)$ can be localized in bounded domains such that it has at most one bound state in each domain, and then the number of the domains is controlled by the Besicovitch covering lemma. Heuristically, since \eqref{eq:CLR-V} is stronger than \eqref{eq:LT-V} (when $d\ge 3$), it is not surprising that the Besicovitch covering lemma is also helpful for the proof of the Lieb--Thirring inequality. However, since \eqref{eq:LT}  seems less linear than \eqref{eq:LT-V} and \eqref{eq:CLR-V}, we think it is still interesting to see that a variant of this microlocal technique works directly for \eqref{eq:LT}. 

%

\section{Proof}

Our proof of \eqref{eq:LT} includes two steps. In the first step, we derive a local version of \eqref{eq:LT} in any ball $B\subset \R^d$ such that  $\int_B \rho_\gamma =2$ (we can replace $2$ by any fixed number bigger than 1). This step is done by combining local versions of the uncertainty and exclusion principles, exactly in the spirit of the proofs in \cite{LunSol-13,LunSol-14,LunNamPor-16} (the only difference is that we are interested in balls rather than cubes). In the second step, we cover $\R^d$ by a suitable collection of balls via the Besicovitch covering lemma, and conclude   \eqref{eq:LT}.

Now we come to the details. Let us denote, with the spectral decomposition $\gamma= \sum_n \lambda_n |u_n\rangle \langle u_n|$ and any ball $B\subset \R^d$, that  
\begin{align} \label{eq:Tr-B-local}
\Tr(-\Delta_B \gamma)  = \sum_n \lambda_n  \int_{B}  |\nabla u_n|^2, \quad \rho_\gamma(x)=\sum_n \lambda_n  |u_n(x)|^2. 
\end{align}
We will always denote by $C_d>0$ a general (large) constant depending only on the dimension $d\ge 1$. The value of $C_d$ may change line by line. 

\bigskip
\noindent
{\bf Step 1:} First, we prove a local version of \eqref{eq:LT}. 

\begin{lemma}[Lieb--Thirring inequality for balls] \label{lem:LT-ball}For any ball $B\subset \R^d$ and any trace class operator $0\le \gamma \le 1$ on $L^2(B)$ satisfying $\int_B \rho_\gamma = 2$, we have
\bq \label{eq:LT-ball}
\Tr(-\Delta_B \gamma) \ge \frac{1}{C_d} \int_{B} \rho_\gamma(x)^{1+\frac{2}{d}} \d x.
\eq
The constraint $\int_B \rho_\gamma = 2$ can be replaced by $a\le \int_B \rho_\gamma \le b$ for any constants $1<a<b<\infty$ (then $C_d$ depends on $d,a,b$). 
\end{lemma}

Lemma \ref{lem:LT-ball} is a consequence of the following well-known results. 

\begin{lemma}[Hoffmann-Ostenhof inequality] \label{lem:HO} For any measurable subset $B\subset \R^d$ and any trace class operator $0\le \gamma \le 1$ on $L^2(B)$ we have
$$
\Tr(-\Delta_B \gamma)  \ge \int_B |\nabla \sqrt{\rho_\gamma}|^2. 
$$
\end{lemma}
\begin{proof} This result was first proved in \cite{HO-77}. For the reader's convenience, here we recall another derivation using the diamagnetic inequality $|\nabla |u| | | \le |\nabla u|$ (see \cite[Theorem 7.21]{LieLos-01}). Indeed, for two complex functions $\{v_n\}_{n=1}^2$, using the diamagnetic inequality for $v_1, v_2$ and $|v_1|+ {\bf i} |v_2|$ (with ${\bf i}^2=-1$) we have 
\begin{align} \label{eq:u1u2}
|\nabla v_1|^2 + |\nabla v_2|^2 &\ge |\nabla |v_1||^2 + |\nabla |v_2||^2 = |\nabla (|v_1| + {\bf i} |v_2|)|^2 \ge \left |\nabla \sqrt{ |v_1|^2 + |v_2|^2}\right|^2.  
\end{align}
Note that \eqref{eq:u1u2} is called the convexity of the gradient which was first published in \cite[Theorem 4.3]{Benguria-79} and \cite[Lemma 4]{BBL-81} (see also \cite[Theorem 7.8]{LieLos-01}). By induction, we find that for any functions $\{v_n\}_{n=1}^\infty$,
\begin{align} \label{eq:u1u2un}
\sum_{n=1}^\infty |\nabla v_n|^2 \ge \left|\nabla \sqrt{ \sum_{n=1}^\infty |v_n|^2}\right|^2.  
\end{align}
We obtain the desired estimate by integrating the pointwise inequality \eqref{eq:u1u2un} with $v_n=\sqrt{\lambda_n} u_n$ and using the definition \eqref{eq:Tr-B-local}. 
\end{proof}

\begin{lemma}[Local uncertainty principle] \label{lem:LU} For any ball $B\subset \R^d$ and any function $g\in H^1(B)$ we have 
$$
\int_B |\nabla g|^2 \ge \frac{1}{C_d} \frac{\int_B |g|^{2(1+ \frac 2 d)}}{(\int_B |g|^2)^{\frac 2 d}} - \frac{C_d}{ |B|^{\frac 2 d}} \int_B |g|^2 .
$$
\end{lemma}

\begin{proof} By scaling we may assume that $B$ is the unit ball. Then by a standard result (see \cite[Theorem 7.41]{Adams-75}), we can find 
$G\in H^1(\R^d)$, an extension of $g$, such that  
$$G_{|B}=g,\quad \|G\|_{H^1(\R^d)} \le C_d \|g\|_{H^1(B)},\quad  \|G\|_{L^2(\R^d)} \le C_d \|g\|_{L^2(B)}.$$
Thus it suffices to show the Gagliardo-Nirenberg type inequality
\begin{align} \label{eq:G-H-U}
\|G\|_{L^{2(1+2/d)}(\R^d)} \le C_d \|G\|_{H^1(\R^d)}^{\frac d {d+2}} \|G\|_{L^2(\R^d)}^{\frac 2 {d+2}},\quad \forall G\in H^1(\R^d),
\end{align}
which is indeed equivalent to Sobolev's inequality $\|G\|_{L^{2(1+2/d)}} \le C_d \|G\|_{H^1}$ (Sobolev's inequality seems weaker, but we can apply it for $G_\ell(x)=\ell^{d/2} G(\ell x)$ and then optimize over $\ell>0$ to obtain \eqref{eq:G-H-U}). See also \cite[Lemma 8]{LunNamPor-16} for another derivation of \eqref{eq:G-H-U} via the fractional Sobolev embedding theorem. 
\end{proof}

\begin{lemma}[Local exclusion principle] \label{lem:LE}For any open ball $B\subset \R^d$ and any trace class operator $0\le \gamma \le 1$ on $L^2(B)$
\bq \label{eq:Neumann-gap}
\Tr(-\Delta_B \gamma) \ge \frac{1}{C_d |B|^{2/d}}  \Big( \int_B \rho_\gamma  - 1\Big).
\eq
\end{lemma}
\begin{proof} This bound goes back to Dyson and Lenard \cite{DysLen-67}. It is a consequence of the Neumann spectral gap 
$$
 -\Delta_B \ge  \frac{1}{C_d |B|^{2/d}} (\1_B- P)  \quad \text{ on }L^2(B)
$$
where $P= |B|^{-1} |\1_B\rangle\langle \1_B|$, the rank-one projection on the unique ground state $|B|^{-\frac 1 2} \1_B$ of the Neumann Laplacian. Then we can take the trace against $\gamma$ and conclude by using  $\Tr (\1_B\gamma)=\int_{B} \rho_\gamma$ and $\Tr (P\gamma)\le \Tr P =1$ (since $0\le \gamma \le 1$).  
\end{proof}

Now we are ready to provide 

\begin{proof}[Proof of Lemma \ref{lem:LT-ball}] Denote $M=\int_B \rho_\gamma$. From Lemmas \ref{lem:HO} and \ref{lem:LU} we have
$$
\Tr(-\Delta_B \gamma)  \ge \int_B |\nabla \sqrt{\rho_\gamma}|^2 \ge \frac{1}{C_d M^{2/d}}  \int_B \rho_\gamma^{1+2/d} - \frac{C_d}{ |B|^{2/d}} M.   
$$
Combining with \eqref{eq:Neumann-gap} we find that for every $\eps>0$, 
$$
 (1+\eps) \Tr(-\Delta_B \gamma)  \ge   \frac{\eps}{C_d M^{2/d}}  \int_B \rho_\gamma^{1+2/d}  + \frac{1}{ |B|^{2/d}} \Big[ \frac{M-1}{C_d} - \eps C_d M\Big].
$$
When $M=2$, we can always choose $\eps= \eps_d >0$ such that 
$$
\frac{M-1}{C_d} - \eps C_d M \ge 0,
$$
and the desired estimate follows. 
\end{proof}

\bigskip
\noindent
{\bf Step 2:} 
In order to put  Lemma \ref{lem:LT-ball} in good use, we need

\begin{lemma}[Besicovitch covering lemma \cite{Besicovitch-45,Morse-47}] Let $E$ be a bounded subset of $\R^d$. Let $\mathcal{F}$ be a collection of balls in $\R^d$ such that every point $x\in E$ is the center of a ball from $\mathcal{F}$. Then there is a sub-collection $\mathcal{G} \subset \mathcal{F}$ such that 
$$
\1_E \le \sum_{B\in \mathcal G} \1_B \le b_d,
$$
namely $E$ is covered by $\bigcup_{B\in \mathcal G} B$ and every point in $E$ belongs to at most $b_d$ balls from $\mathcal{G}$. Here the constant $b_d>0$  depends only on the dimension $d\ge 1$ (in particular, $b_d$ is independent of $E$).
\end{lemma}

Now we conclude 

\begin{proof}[Proof of the Lieb--Thirring inequality \eqref{eq:LT}] By a density argument it suffices to assume that $\gamma= \sum_{i=1}^N \lambda_n |u_n\rangle \langle u_n|$ with $N<\infty$, with orthonormal functions $\{u_n\}_{n=1}^N \subset C_c^\infty(\R^d)$, and with $0\le \lambda_n \le 1$. 

If $\int_{\R^d}\rho_\gamma\le 2$, then from Lemma \ref{lem:HO} and Sobolev's inequality \eqref{eq:S} we find that 
$$
\Tr (-\Delta \gamma) \ge \int_{\R^d} |\nabla \sqrt{\rho_\gamma}|^2 \ge \frac{\int_{\R^d} \rho_\gamma^{1+\frac 2 d}}{ C_d (\int_{\R^d} \rho_\gamma)^{2/d}} \ge \frac{1}{C_d 2^{\frac 2 d}}\int_{\R^d} \rho_\gamma^{1+\frac 2 d}. 
$$
Hence, it remains to consider the case 
$\int_{\R^d} \rho_\gamma \ge 2.$ Since $\rho_\gamma$ is continuous and its support $E\subset \R^d$ is bounded, for every $x\in E$ we can find a ball $B_x \subset \R^d$ centered at $x$ such that 
\bq  \label{eq:choice-B}
\int_{B_x} \rho_\gamma = 2. 
\eq
Applying the Besicovitch covering lemma, from the collection of balls $\mathcal{F}=\{B_x\}_{x\in E}$, we can find a sub-collection $\mathcal{G} \subset \mathcal {F}$ such that 
\bq  \label{eq:E-subset-G}
\1_E \le \sum_{B\in \mathcal G} \1_B \le b_d.
\eq
The second bound in \eqref{eq:E-subset-G} implies that
\begin{align*} 
b_d \int_{\R^d} |\nabla u_n|^2 \ge  \sum_{B\in \mathcal{G}} \int_{B} |\nabla u_n|^2
\end{align*} 
for all $n$. Therefore, 
\begin{align*} 
b_d \Tr (-\Delta \gamma) &=  \sum_{n=1}^N \lambda_n b_d \int_{\R^d} |\nabla u_n|^2 \ge   \sum_{n=1}^N  \lambda_n \sum_{B\in \mathcal{G}}  \int_{B} |\nabla u_n|^2\\
&= \sum_{B\in \mathcal{G}}  \sum_{n=1}^N \lambda_n  \int_{B} |\nabla u_n|^2   = \sum_{B \in \mathcal G} \Tr (-\Delta_{B} \gamma). 
\end{align*} 
On the other hand, for every ball $B \subset \mathcal G$, we have $\int_{B} \rho_\gamma = 2$ and $0\le \gamma \le 1$ on $L^2(B)$ (since $0\le \gamma \le 1$ on $L^2(\R^d)$). Thus Lemma \ref{lem:LT-ball} gives 
\bq 
\Tr(-\Delta_B \gamma) \ge \frac 1 C_d \int_{B} \rho_\gamma ^{1+\frac{2}{d}}.
\eq
Finally, combining the local bound \eqref{eq:E-subset-G} with the covering property \eqref{eq:E-subset-G} we conclude that
\bq 
b_d \Tr (-\Delta \gamma) \ge \sum_{B \in \mathcal G} \Tr (-\Delta_{B} \gamma) \ge  \frac 1 {C_d}  \sum_{B \in \mathcal G} \int_{B} \rho_\gamma^{1+2/d} \ge    \frac 1 {C_d} \int_{\R^d} \rho_\gamma^{1+2/d}.
\eq
The proof is complete. 
\end{proof}

\end{document}